\documentclass{amsart}
 
\usepackage{amsthm} 
\usepackage{amsmath} 
\usepackage{amssymb} 
\usepackage[bookmarks]{hyperref} 
\usepackage{paralist} 
\usepackage[numbers]{natbib} 
\usepackage{verbatim} 
\usepackage{url} 
 
	\usepackage{soul}
        \usepackage{color}
	
	\definecolor{lightblue}{rgb}{.60,.60,1}
	
        \definecolor{brown}{rgb}{.5,.5,.5}

\newcommand{\C}{\mathcal{C}}
\newcommand{\N}{\mathbf{N}}
\newcommand{\R}{\mathbf{R}}
\newcommand{\G}{\mathcal{G}}
\renewcommand{\S}{\mathcal{S}}
\newcommand{\Q}{\mathcal{Q}}
\newcommand{\Le}{\mathcal{L}^1}
\newcommand{\lcode}{\left\langle}
\newcommand{\rcode}{\right\rangle}
\newcommand{\pair}[1]{{\lcode#1\rcode}}
\newcommand{\size}[1]{{\left|#1\right|}}

\newcommand{\intersect}{\mathrel{\cap}}

\newcommand{\restr}{\mathrel{\upharpoonright}}
\DeclareMathOperator{\density}{density}

\DeclareMathOperator{\card}{\mathsf{card}}

\DeclareMathOperator{\entropy}{entropy}
\newcommand{\Ha}{\mathcal{H}}
\newcommand{\cHa}{\mathsf{c}\mathcal{H}}
\newcommand{\M}{\mathcal{M}}
\newcommand{\Eleq}{\mathcal{E}_{\leq\alpha}}
\newcommand{\Eeq}{\mathcal{E}_{=\alpha}}
\newcommand{\Elt}{\mathcal{E}_{<\alpha}}

\DeclareMathOperator{\dimh}{\dim_\mathrm{H}}
\newcommand{\dimC}{\dimh\C}

\DeclareMathOperator{\cdim}{cdim}
\DeclareMathOperator{\cdimh}{\cdim_\mathrm{H}}
\DeclareMathOperator{\cdimp}{\cdim_\mathrm{P}}

\DeclareMathOperator{\dimbu}{\overline{\dim}_\mathrm{B}}
\theoremstyle{plain}        \newtheorem{thm}{Theorem}[section]
\theoremstyle{definition}   \newtheorem{defn}[thm]{Definition}
\theoremstyle{plain}        
\theoremstyle{plain}        
\theoremstyle{plain}        \newtheorem{cor}[thm]{Corollary}
\theoremstyle{plain}        \newtheorem{lemma}[thm]{Lemma}
\theoremstyle{remark}       
\theoremstyle{remark}       \newtheorem*{rem}{Remark}
\theoremstyle{plain}        \newtheorem{ques}[thm]{Question}
\theoremstyle{plain}        
\theoremstyle{plain}        \newtheorem*{lorentzlemma}{Lorentz's~Lemma}
\theoremstyle{plain}        \newtheorem*{marstrandlemma}{Marstrand's~Lemma}

\numberwithin{equation}{section}

\begin{document}

\title{Translating the Cantor set by a random}%real}
%On the dimension of the intersection of a random translations of the Cantor set with sets of a given constructible dimension}

\author{Randall Dougherty}
\address{CCR--La Jolla, 4320 Westerra Court, San Diego,
  CA 92121}
\email{rdough@ccrwest.org}

% \and

\author{Jack Lutz}
\address{Iowa State University, Ames, IA 50011}  
\email{lutz@cs.iastate.edu} 
\thanks{Research supported by NSF Grants 0652569 and 0728806.} 

%\and

\author{R. Daniel Mauldin}
\address{University of North Texas, Denton,TX 76203 }
\email{mauldin@unt.edu} 
\thanks{Research supported by NSF Grant DMS-0700831.}

%\and

\author{Jason Teutsch}
\address{Ruprecht-Karls-Universit\"at Heidelberg, D-69120 Heidelberg}
\email{teutsch@math.uni-heidelberg.de}
\thanks{Research supported by
Deutsche Forschungsgemeinschaft grant ME 1806/3-1.}


%\begin{document}

\begin{abstract}
We determine the constructive dimension of points in random translates of the Cantor set.  The Cantor set ``cancels randomness'' in the sense that some of its members, when added to Martin-L\"{o}f random reals, identify a point with lower constructive dimension than the random itself.  In particular, we find the Hausdorff dimension of the set of points in a Cantor set translate with a given constructive dimension.
\end{abstract}

\maketitle

\section{Fractals and randoms}
We explore an essential interaction between algorithmic randomness, classical fractal geometry, and additive number theory.  In this paper, we consider the dimension of the intersection of a given set with a translate of another given set.  We shall concern ourselves not only with classical Hausdorff measures and dimension but also the effective analogs of these concepts.

More specifically, let $\C$ denote the standard middle third Cantor set \cite{Fal03, Zie04}, and for each number $\alpha$ let
\begin{equation} \label{def: Eeq}
\Eeq = \{x : \cdimh \{x\} = \alpha \}
\end{equation}
consist of all real numbers with constructive dimension $\alpha$. We answer a question posed to us by Doug Hardin by proving the following theorem: if $1 -\log2/\log3 \leq \alpha \leq 1$ and $r$ is a Martin-L\"{o}f random real, then the Hausdorff dimension of
\begin{equation} \label{pong}
(\C+r) \cap \Eeq
\end{equation}
is $\alpha -(1 -\log 2/\log 3)$.  From this result we obtain a simple relation between the effective and classical Hausdorff dimensions of \eqref{pong}; the difference is exactly $1$ minus the dimension of the Cantor set.  We conclude that many points in the Cantor set additively cancel randomness.

We discuss some of the notions involved in this paper.   Intuitively, a real is ``random'' if it does not inherit any special properties by belonging to an effective null class.  We say a number is \emph{Martin-L\"of random} \cite{DH10, ML66} if it ``passes'' all Martin-L\"of tests.  A \emph{Martin-L\"of test}  is
a uniformly computably enumerable (c.e.) sequence \cite{DH10, Soa87} of open sets
$\{U_m\}_{m\in\N}$ with  $\lambda(U_m) \leq 2^{-m}$, where $\lambda$ denotes
Lebesgue measure \cite{Zie04}. A number $x$ \emph{passes} such a test if $x \not\in
\cap_m U_m$.

The \emph{Kolmogorov complexity} of a string $\sigma$, denoted $K(\sigma)$, is the length (in this paper we will measure length in ternary units) of the shortest program (under a fixed universal machine) which outputs $\sigma$ \cite{LV08}.  For a real number $x$, $x\restr n$ denotes the first $n$ digits in a ternary expansion of $x$.  Martin-L\"{o}f random reals have high initial segment complexity \cite{DH10}; indeed every Martin-L\"{o}f random real $r$ satisfies $\lim_n K(r \restr \nolinebreak n)/n = 1$.  This fact conforms with our intuition that random objects do not compress much.

% Roughly speaking, the dimension of a set is the amount of space it fills in Euclidean space \cite{Fal03}.
We introduce a couple of classical dimension notions.  Let $E \subseteq \R^n$.  The \emph{diameter} of $E$, denoted $\size{E}$,  is the maximum distance between any two points in $E$. We will use $\card$ for cardinality.  A \emph{cover} $\G$ for a set $E$ is a collection of sets whose union contains $E$, and $\G$ is a $\delta$-\emph{mesh} cover if the diameter of each member $\G$ is at most $\delta$.  For a number $\beta \geq 0$, the \emph{$\beta$-dimensional Hausdorff measure} of $E$, written $\Ha^\beta(E)$, is given by $\lim_{\delta \to 0} \Ha^\beta_\delta(E)$ where 
\begin{equation} \label{def: Ha}
\Ha^\beta_\delta(E) = \inf\left\{\sum_{G \in \G} \size{G}^\beta : \text{$\G$ is a countable $\delta$-mesh cover of $E$} \right\}.
\end{equation}
The \emph{Hausdorff dimension} of a set $E$, denoted \emph{$\dimh(E)$}, is the unique number $\alpha$ where the $\alpha$-dimensional Hausdorff measure of $E$ transitions from being negligible to being infinitely large; if $\beta < \alpha$, then $\Ha^\beta(E) = \infty$ and if $\beta > \alpha$, then $\Ha^\beta(E) = 0$ \cite{Fal03, Zie04}.  Let $S_\delta(E)$ denote the smallest number of sets of diameter at most $\delta$ which can cover $E$.  The  \emph{upper box-counting dimension} \cite{Fal03} of $E$ is defined as
\begin{equation*}
\dimbu(E) = \limsup_{\delta \to 0} \frac{\log S_\delta(E)}{-\log \delta}.
\end{equation*}

The \emph{effective} (or \emph{constructive}) \emph{$\beta$-dimensional Hausdorff measure} of a 
set $E$, $\cHa^\beta(E_k)$, is defined exactly in the same way as
Hausdorff measure with the restriction that the covers be uniformly
c.e.\ open sets \cite[Definition~13.3.3]{DH10}. This yields the corresponding notion of the
\emph{effective} (or \emph{constructive}) \emph{Hausdorff dimension} of a set $E$,
$\cdimh E$.
% I used constructive dimension to refer to points and
% effective Hausdorff dimension to refer to sets.
%The \emph{effective dimension} of a singleton $\{x\}$ is $\cdimh(\{x\})$.
% dim and Dim were nowhere else used in the paper so I deleted them (JT).
Lutz \cite{Lut03} showed
\begin{equation} \label{meatball}
\cdimh E= \sup \{\cdimh \{x\} : x \in E\},
\end{equation}
and from work of Mayordomo~\cite{May02}($\ge$) and Levin~\cite{Lev73}($\leq$) (also see \cite{DH10}) we have for any real number $x$,
\begin{equation} \label{spaghetti}
\cdimh \{x\} = \liminf\limits_{n\to\infty} \frac{K(x \restr n)}{n}.
\end{equation}
Mayordomo and Levin prove their results in $\{0,1\}^\N$, but the results
carry over to the reals.
We define the \emph{constructive dimension} of a point $x$ to be the effective Hausdorff dimension of the singleton $\{x\}$.  A further effective dimension notion, the \emph{effective packing dimension} \cite{KHLM07, DH10} satisfies
\begin{itemize}
\item $\cdimp \{x\} = \limsup\limits_{n\to\infty} \frac{K(x \restr n)}{n}$, and 
\item $\cdimp \C = \dimC$.
\end{itemize}
Let us first note some simple bounds on the complexity of a point in the translated Cantor set $\C+r = \{y: (\exists x\in \C)\: [y = x+r]\}$.
\begin{thm} \label{simp}  Let $x \in \C$ and let $r$ be a
  Martin-L\"{o}f random real. Then
$$
 1- \dimC \leq \cdimh\{x+r\} \leq 1.
$$
\end{thm}

\begin{proof} Given $x$ and $r$, there is a constant $c$ such that for
  all $n$, 
 \[
 K[(x+r) \restr n] + K(x\restr n) \geq K(r \restr n) +c.
 \]
Thus,
\begin{align*}
1 \geq \cdimh \{x+r\} &= \liminf\limits_{n\to\infty}\frac{K[(x+r )\restr n]}{n} \\
& \geq \liminf\limits_{n\to\infty}\frac{K(r \restr n)}{n} - \limsup\limits_{n\to\infty} \frac{K(x \restr n)}{n} \\
& \geq 1 - \cdimp\{x\} \geq 1-\cdimp\C =  1-\dimC. \mbox{\qedhere}
\end{align*}
%This concludes the argument.
\end{proof}

In Section~\ref{sec: spcr}, we will indicate how some points cancel randomness. We
show that for every $r$ there exists an $x$ such that the constructive dimension of $x+r$ is
as close to the lower bound as one likes. Later we will show that
for each $r$ and number $\alpha$ within the correct bounds, not only
does there exist some $x \in \C$ so that $x+r$ has constructive dimension $\alpha$, but we
will determine the Hausdorff dimension of the set of all $x$'s with constructive dimension $\alpha$.  At this point let us give a heuristic argument indicating what the Hausdorff dimension of this set might be.

Fix a number $1 - \dimC < \alpha < 1$, and  following the notation in \eqref{def: Eeq}, let
\[
\Eleq = \{x : \cdimh\{x\} \leq \alpha \}.
\]
From \cite{Lut03} (see also \cite{CH94}), we know that the effective Hausdorff dimension of $\Eleq$ satisfies $\cdimh\Eleq = \dimh\Eleq = \alpha$.  Since
the upper box counting dimension of $\C$ satisfies
$\dimbu \C = \dimC$  \cite[Example~3.3]{Fal03}, we have $\dimh(\C\times \Eleq) =
\dimC + \alpha$ \cite[Corollary~7.4]{Fal03}.
Define $f:\R^2 \mapsto \R^2$ by
\begin{equation} \label{eq: defn of f}
f(x,y) = (y-x,y).
\end{equation}
Then $f$ is a bi-Lipschitz map and therefore preserves Hausdorff dimension \citep[Lemma~1.8]{Fal86}. So, letting $B = f^{-1}(\C\times \Eleq)$, we have $\dimh B = \dimC + \alpha$.  The \emph{vertical fiber} of $B$ at $x$, or set of points $y$ such that $(x,y) \in B$, is
\begin{equation} \label{eqn: defn of B_x}
B_x = (\C+x)\intersect \Eleq.
\end{equation}
Let $\gamma > \dimC + \alpha$. By the Fubini type inequality for
Hausdorff measures \cite[Theorem~5.12]{Fal86}, there is a positive constant $b$
such that
\[
 0 
 = {\Ha}^\gamma(B) \geq b \int {\Ha}^{\gamma-1}(B_x)\ d{\Ha}^1(x)
 = b\int {\Ha}^{\gamma-1}(B_x)\ dx.
\]
So for Lebesgue measure a.e.\ $x$, $\mathcal{H}^{\gamma -1}[(\C+x)\intersect
\Eleq] = 0$. Therefore, for Lebesgue measure a.e.\ $x$,
\[
\dimh[(\C+x) \intersect \Eleq] \leq \alpha -(1 -\dimC).
\]
We would like to turn this inequality into an equality for every
Martin-L\"{o}f random real $x$, but even showing that inequality holds
for all Martin-L\"of randoms is a problem. This is because, in general, if one has a non-negative Borel measurable function $f$ and $\int f(x)\ dx = 0$, then $f(x) = 0$ for Lebesgue measure almost every $x$, but there may be Martin-L\"{o}f random $x$'s for which $f(x) > 0$. In Section~\ref{sec: upper bound} of this paper, we more carefully analyze our particular situation to obtain the conjectured upper bound.

\section{Some points cancel randomness} \label{sec: spcr}
We begin with a simple example illustrating how points in the Cantor set can counteract randomness.  Let us briefly review some facts about the standard middle-third
Cantor set.  
\begin{enumerate}
\item We may express any $x \in [0,1]$ as a ternary expansion:
\[
x = .x_1x_2x_3\ldots = \sum_{n=1}^\infty \frac{x_n}{3^n}
\]
where each $x_n \in \{0,1,2\}$.
The Cantor set $\C$
consists of those $x$ for which the $x_n$'s are all 0 or 2,
and the half-size Cantor set $\frac{1}{2}\C$
consists of those $x$ for which the $x_n$'s are all 0 or 1.

\item Any number in the interval $[0,2]$ can be written as a sum of two elements of the Cantor set.  Indeed $\frac{1}{2}\C + \frac{1}{2}\C = [0,1]$ because the coordinates of any ternary decimal can be written as $0+0$, $0+1$, or $1+1$.

\item The Hausdorff dimension and effective Hausdorff dimension of the Cantor set agree (see \citep[Theorem~1.4]{Fal86} and \citep[Section~1.7.1]{Rei04}):
\[
\dimC = \cdimh\C = \frac{\log 2}{\log 3} \approx 0.6309.
\]
All the usual notions of dimension: Hausdorff, packing, upper and lower Minkowski or box counting, agree on $\C$ \cite{Fal03}.
\end{enumerate}

Since the Cantor set contains the point 0, it is immediate that $\C+r$ contains points of constructive dimension 1 whenever $r$ is Martin-L\"{o}f random.  We now present a simple construction which identifies points with lower constructive dimension.

\subsection{A point within 2/3 of optimal}
Let $r \in [0,1]$ be a real with ternary expansion $.r_1 r_2 \dotsc$.
Choose $t = .t_1 t_2 \dotsc \in \C$ as follows.  Let
\[
t_n =
\begin{cases}
0 & \text{if $r_n \in \{1,2\}$,} \\
2 & \text{otherwise.}
\end{cases}
\]
Then
\[
r_n + t_n =
\begin{cases}
2 & \text{if $r_n = 0$,} \\
1 & \text{if $r_n = 1$,} \\
2 & \text{if $r_n = 2$,}.
\end{cases}
\]
Since $(0,\frac{1}{3}, \frac{2}{3})$ is the limiting frequency
probability vector for this sequence, the constructive dimension of
this sequence is dominated by the effective Hausdorff dimension of the set of
all sequences with this limiting frequency vector.  By \cite[Lemma~7.3]{Lut03}, we have
\[\cdimh\{r + t\} \leq
\entropy \left(0, \frac{1}{3}, \frac{2}{3} \right) = -\frac{1}{3} \log_3 \frac{1}{3} - \frac{2}{3} \log_3\frac{2}{3} = 1 - \frac{2}{3} \cdot \dimC.
\]

This shows that for every $r$, there exists some point in $\C+r$ whose constructive dimension is at most $1 - (2/3)\dimC$.  Next we construct points whose constructive dimensions approach the $1 -\dimC$ limit given in Theorem~\ref{simp}.

\subsection{Building blocks: achieving near the limit}
We consider a more refined example.  Recall that $\frac{1}{2}\C$ is the set of all ternary decimals in $[0,1]$ made from 0's and 1's (and no 2's), and take $\frac{1}{2}E_3$ to be the set of ternary decimals in $[0,1]$ generated from concatenated blocks in
\[
B_3 = \{000, 002, 021, 110, 112\}.
\]
So in particular $\frac{1}{2}\C$ is generated by concatenating the blocks
\begin{equation} \label{eq: A3}
C_3 = \{000, 001, 010, 011, 100, 101, 110, 111\}.
\end{equation}
By exhaustion, any ternary block of length 3 can be written as the sum of a member of $C_3$ plus a member of $B_3$ (e.g.\ $020 = 002 + 011$).  Therefore $\frac{1}{2}\C + \frac{1}{2}E_3 = [0,1]$, and furthermore, as we shall see in \eqref{eq: dim E_n},
\[
\cdimh E_3 \leq \frac{\log 5}{\log 27} \approx 0.4883.
\]

The following are examples of optimal complementary block sets for each length (in terms of size).  These blocks are not unique: for each length $k$, there is more than one smallest block set which can be added to the length $k$ analogue of \eqref{eq: A3} in order to achieve all ternary numbers up to length $k$.
\begin{align*}
B_1 &= \{0,1\}, \\
B_2 &= \{00, 02, 11\}, \\
B_3 &= \{000, 002, 021, 110, 112\}, \\
B_4 &= \{0000, 0002, 0011, 0200, 0202, 0211, 1100, 1102, 1111\}, \\
\begin{split}
B_5 &= \{00000, 00002, 00021, 00112, 00210, 01221, 02012, \\ &\quad 02110, 02201, 10212, 11010, 11101, 11120, 11122\}.
\end{split}
\end{align*}
Note that $B_4$ is just the product $B_2 \times B_2$ and is still optimal.  We wonder whether products can be optimal for larger indices as well.

 A set $E \subseteq \R$ is called \emph{computably closed} if there
 exists a computable predicate $R$ such that $x \in E \iff (\forall
 n)\: R(x \restr n)$. We shall use the following combinatorial lemma of Lorentz to prove that there exist sufficiently small complementary blocks for each length whose members can be concatenated to achieve computably closed sets with low effective Hausdorff dimension (Theorem~\ref{thm: finite blocks}).
\begin{lorentzlemma}[\citep{Lor54}]
There exists a constant $c$ such that for any integer $k$, if $A
\subseteq [0, k)$ is a set of integers with $\size{A} \geq \ell \geq
2$, then there exists a set of integers $B \subseteq (-k,k)$ such that
$
A + B \supseteq [0, k)
$
with $\size{B} \leq ck\frac{\log \ell}{\ell}$.
\end{lorentzlemma}
Although Lorentz's~Lemma as such does not appear explicitly in Lorentz's original paper, as mentioned in \cite{EKM81}, his argument in \citep[Theorem~1]{Lor54} proves the statement above. 
\begin{thm} \label{thm: finite blocks}
There exists a uniform sequence of computably closed sets $E_1, E_2, \dotsc$ such that
\begin{enumerate}[\scshape (i)]
\item $\frac{1}{2}\C + \frac{1}{2}E_n = [0,1]$ for all $n$, and

\item $\lim_{n \to \infty} \cdimh E_n = 1 - \dimC$.
\end{enumerate}
\end{thm}
\begin{proof}
For $k>0$, let 
\[
I_k = \{i : 0 \leq i < 3^k\},
\]
and let
\begin{align*}
C_k &= \{i : \text{$0 \leq i < 3^k$ and $i$ has only 0's and 1's in its ternary expansion}\} \\
&= \left\{\sum_{j=0}^{k-1} \delta_j  3^j : \delta_j \in \{0,1\}\right\}.
\end{align*}
By Lorentz's~Lemma (applied to $C_k$), there exists a set $B_k \subseteq (-3^k, 3^k)$ with
\begin{equation} \label{eq: S_k sec2}
I_k \subseteq C_k + B_k
\end{equation}
satisfying
\[
\card{B_k} \leq c' \cdot 3^k \cdot \frac{\log \left(2^k\right)}{2^k} = c'k\log 2 \left(\frac{3}{2}\right)^k = ck \left(\frac{3}{2}\right)^k,
\]
where $c'$ is the constant from Lorentz's Lemma and $c = c'\log 2$.  Set
\[
\frac{1}{2} \C = \left\{ \sum_{n=1}^\infty \frac{a_n}{3^{kn}} : a_n \in C_k \right\} \quad\text{and}\quad
\frac{1}{2} E_k = \left\{ \sum_{n=1}^\infty \frac{b_n}{3^{kn}} : b_n \in B_k \right\}.
\]
Let $x \in [0,1]$ have ternary expansion 
\[
0.x_1 x_2 x_3 \dotsc =
\sum_{n=1}^\infty\sum_{j=1}^k \frac{x_{(n-1)k+j}}{3^{(n-1)k+j}} =
\sum_{n=1}^\infty\sum_{s=0}^{k-1}\frac{x_{nk-s}}{3^{nk-s}} =
\sum_{n=1}^\infty\frac{1}{3^{nk}}\left(\sum_{s=0}^{k-1}x_{nk-s}3^s\right).
\]
  By 
\eqref{eq: S_k sec2}, there exist sequences $\{a_n\}$ with members in
$C_k$ and $\{b_n\}$ from $B_k$ such that for all $n \geq 1$,
\[
\sum_{s=0}^{k-1} x_{nk-s} 3^s = a_n + b_n,
\]
and therefore
\[
x = \sum_{n=1}^\infty \frac{a_n}{3^{kn}} +  \sum_{n=0}^\infty \frac{b_n}{3^{kn}} \in \frac{1}{2}\C + \frac{1}{2} E_k,
\]
which proves part~\textsc{(i)}.

Define
\begin{equation*}
\gamma_k 
= \frac{\log (\card B_k)}{\log 3^k} \
\leq 1 - \frac{\log 2}{\log 3} + \frac{\log c + \log k}{k\log 3}.
\end{equation*}
To prove part~\textsc{(ii)}, we first note that $\cdimh(E_k) \leq
\gamma_k$.  For every $n>0$, we can uniformly cover $E_k$ with $(\card
B_k)^n$ intervals of size $3 \cdot 3^{-kn}$.  Indeed, there are $\card B_k$
choices for each of the first $n$ blocks in any member of $E_k$, and a
closed interval of length $3 \cdot 3^{-kn}$ covers all possible extensions of
each such prefix.  Each $E_k$ is a computably closed set and we have:
\begin{equation}  \label{eq: dim E_n}
\cHa^{\gamma_k}(E_k)
\leq \lim_{n \to \infty} \left(\card B_k\right)^n \cdot 3^{\gamma_k} \cdot (3^{-kn})^{\gamma_k}
= 3^{\gamma_k}.
\end{equation}
So, $\limsup_{k \to \infty} \dimh{E_k} \leq 1 - \dimh\C$.
Also, we have $\gamma_k \ge \dimbu(E_k)$.
Again, applying the fact the Lipschitz map $(x,y) \mapsto x+y$ doesn't increase dimension \citep[Lemma~1.8]{Fal86} together with a bound on the dimension of a product set in terms of the dimension of its factors (\cite[Product~formula~7.3]{Fal03}) we have 
\begin{equation} \label{eq: zxcvb}
1 =\dimh(\C+E_k) \leq  \dimh(\C \times E_k) \leq \dimbu\C+ \dimh E_k = \dimh\C+ \dimh E_k.
\end{equation}
The leftmost equality of \eqref{eq: zxcvb} follows from part~(\textsc{i}) and the rightmost equality follows from \cite[Example~3.3]{Fal03}.  Thus part~\textsc{(ii)} holds.
\end{proof}

From the construction of the set $E_k$ one would think that $\dimh E_k = \gamma_k$ and $0 < \Ha^{\gamma_k}(E_k) < \infty$.  However, it is not clear that the similarity maps that one might naturally use to generate the self-similar set $E_k$ satisfy the open set condition, see \cite{Fal86}.  In fact, there are possible cases (e.g., when $B_k$ contains two consecutive numbers and two numbers that differ by $3^k$) where we would get $\dimh E_k < \gamma_k$.

We obtain immediately from Theorem~\ref{thm: finite blocks} the following:
\begin{cor} \label{cor: finite blocks}
For every real $r \in [0,2]$ and every $\epsilon > 0$, there exists a point in $\C+r$ whose constructive dimension is less than $1 - \dimC + \epsilon$.
\end{cor}
\begin{proof}
Let $E_n$ be as in Theorem~\ref{thm: finite blocks} with $n$ large enough so that $\cdimh E_n < 1 - \dimC + \epsilon$,  and let $r \in[0,2]$.  Then $r' = 2-r \in [0,2]$, and there are points $x\in \C$ and $y\in E_n$ such that $x+y = 2-r$.  Thus $x+r \in 2-E_n$; hence 
\[
\cdimh\{x+r\} \leq \cdimh(-E_n+2) = \cdimh E_n < 1 - \dimC + \epsilon
\]
as desired.
\end{proof}
As we shall see in Section~\ref{sec: lb}, we can even achieve a closed set $E$ of effective Hausdorff dimension $1-\dimC$ satisfying $\frac{1}{2}\C + \frac{1}{2}E = [0,1]$.

\section{Lower bound} \label{sec: lb}
In Section~\ref{sec: spcr} we demonstrated the existence of points in the Cantor set which cancel randomness; we now show there are many such points.  Instead of searching for individual points with small dimension, we now characterize the Hausdorff dimension (and effective Hausdorff dimension) of all such points.  We use Lorentz's~Lemma again to upgrade Theorem~\ref{thm: finite blocks} and Corollary~\ref{cor: finite blocks}.  Our upgrade proceeds in two phases.  The second phase occurs later in Section~\ref{sec: lb2} as it relies on the upper bound results from Section~\ref{sec: upper bound}.  Our procedure is the same as that used in \cite{EKM81}.

\begin{defn}
The \emph{density} of a set $A = \{k_1 < k_2 <k_3 < \ldots\} \subseteq \N$ is defined to be
\[
\density(A) = \lim_{n\to \infty}\frac{\card(A \intersect \{1,2, \dotsc, n\})}{n},
\]
provided this limit exists.
\end{defn}
We note that $\density(A) = \lim_{n\to \infty}\frac{n}{k_n}$.  Below $A[n]$ will denote the length $n$ prefix of $A$'s characteristic function and $\lfloor x \rfloor$ is the integer part of $x$.

\begin{thm} \label{lbt}  Let $1 - \dimC \leq \alpha \leq 1$ and let $r \in [0,1]$. Then
\[
\dimh \left[(\C+r) \intersect \Eleq \right] \ge \alpha - 1 + \dimC.
\]
\end{thm}

\begin{proof}
For $\alpha = 1$, the theorem clearly holds.  Thus assume
\[
D = \frac{1-\alpha}{\dimC} > 0,
\]
and define
\begin{equation} \label{eqn: density(A)}
A = \left\{\left\lfloor y/D \right\rfloor : y \in \N\right\}
\end{equation}
so that $D = \density(A)$.  Let $\C_A$ ($\C_{\bar{A}}$) be the set of $x \in [0,1]$ having a ternary expansion whose digits are all $0$ or~$2$, where the $2$'s only occur at positions in~$A$ (positions not in~$A$).

Now
\[
\dimbu \C_A \leq \limsup_{n \to \infty} \frac{\log 2^n}{\log 3^{k_n -1}} = \dimC \cdot \limsup_{n \to \infty} \frac{n}{k_n-1} = 1 - \alpha.
\]
Since upper box counting dimension dominates Hausdorff dimension \cite[p.~43]{Fal03}, we also have
$\dimh \C_A \leq \dimbu \C_A \leq 1-\alpha$.  As in \eqref{eq: zxcvb}, the Lipschitz map $(x,y) \mapsto x+y$ does not increase dimension \cite[Lemma~1.8]{Fal86}, so $\dimh(\C_A + \C_{\bar{A}}) \leq \dimh(\C_A \times \C_{\bar{A}})$.  Since $\C = \C_A + \C_{\bar{A}}$, it follows from \cite[Product~formula~7.3]{Fal03} that
\begin{equation} \label{eqn: hello kitty}
\dimh \C_{\bar{A}} \geq \dimC - \dimbu \C_A \geq \dimC + \alpha - 1.
\end{equation}

We pause from the main argument to prove the following two lemmas.  First we exploit the special form of the set $A$.
\begin{lemma} \label{lem: K(A) is small}
 $K(A[n]) \leq 4\log_3 n + O(1)$.
\end{lemma}
\begin{proof}[Proof of Lemma~\ref{lem: K(A) is small}]
%If $D=0$, then $A = \emptyset$ and the lemma is immediate. So assume $D >0$.   To compute $A[n]$, it suffices to compute $\lfloor y/D \rfloor$ for every integer $y \leq \lceil nD \rceil \leq n$.  We now investigate how may bits of $D$ we need in order to compute $\lfloor y/D \rfloor$ for all $y \leq n$.
Let $\frac{r}{s}$ be the largest fraction with $s \leq n$ such that
$\frac{r}{s} \leq \frac{1}{D}$. Notice if we know $r$, $s$, and $n$, we
can compute $A[n]$ because
\begin{equation*}\label{happy}
\left\lfloor \frac{y}{D}\right\rfloor = \left\lfloor \frac{ry}{s}\right\rfloor
\end{equation*}
for $1 \leq y \leq n$.  (To see this, notice that if $x = \lfloor\frac{y}{D}\rfloor >
\frac{ry}{s}$, then $\frac{ry}{s} < x \leq \frac{y}{D}$. This would
give us $\frac{r}{s} < \frac{x}{y} \leq \frac{1}{D}$, contradicting
maximality of $\frac{r}{s}$.)  Specifying $r$, $s$, and $n$ requires a ternary string of length at most
\begin{equation} \label{lala}
\log_3\left(\frac{n}{D}\right) + \log_3 n + \log_3 n + (2\log_3 \log_3 n + 1) +O(1),
\end{equation}
where the ``$2\log_3 \log_3 n + 1$'' bits are used to mark the ends of the ``$\log_3 n$'' bit strings, and $O(1)$ tells the universal machine how to process the input.  The lemma now follows by noting that $4\log_3 n + O(1)$ is an upper bound for \eqref{lala}.
%%Thus
%\[
%K(A[n]) \leq 4\log_3 n + O(1). \qedhere
%\]
\end{proof}
%Let $D_s$ be the ternary number consisting of the first $s$ ternary digits of $D$ followed by all zeros.  Then for all sufficiently large $s$,
%\[
%\frac{1}{D_{s+1}} - \frac{1}{D} = \frac{D - D_{s+1}}{D_{s+1} \cdot D} \leq \frac{3^{-s}}{D^2/2},
%\]
%and so $y/D_{s+1} - y/D \leq y/3^{s-d}$ where $d$ is the positive constant $-\log_3(D^2/2)$.  In particular
%\[
%\frac{y}{D_{2\log_3 \left(\frac{y}{D}+1\right)+d+1}} - \frac{y}{D} \leq 3^{-\log_3 \left(\frac{y}{D}+1\right)},
%\]
%and thus
%\begin{equation} \label{happydigit}
%K\left[\frac{y}{D} \restr \log_3 \left(\frac{y}{D}\right) \right] \leq K\left(D \restr \left[2\log_3\left(\frac{y}{D}+1\right) + d + 1\right]\right) + O(1).
%\end{equation}
%Note that for any $y$, the integer part of $\frac{y}{D}$ can be determined from the first $\log_3 \left(\frac{y}{D}\right)$ digits of $\frac{y}{D}$'s ternary representation.  Combining this observation with \eqref{happydigit}, we obtain
%\begin{multline*}
%K(A[n]) \leq K\left[\left\{ \frac{y}{D} \restr \log_3 \left(\frac{y}{D}\right): y \leq (n+1)D  \right\} \right] + O(1) \\
%\leq  K\left(D \restr \left[2\log_3\left(n+2\right)+ d + 1\right]\right) + O(1)
%\leq 2\log n + a
%\end{multline*}
%for some constant $a$.

\begin{lemma} \label{lem: lbt dim E lemma}   
There exists a closed set $E$ such that $\cdimh E \leq \alpha$ and $\C_A + E = [0,2]$.
\end{lemma}
\begin{proof}[Proof of Lemma~\ref{lem: lbt dim E lemma}]
We follow the outline of our prior argument from Section~\ref{sec: spcr}.  The idea is to take $E$ to
be a set generated by concatenating elements from the blocks $B_1,
B_2, B_3, \dotsc$ as in Theorem~\ref{thm: finite blocks}.

For each $k > 0$,  let $m_k = k^2$, let $n_k$ denote the difference $m_k - m_{k-1}$, and let
\[
I_k = \{ i : 0 \leq i < 3^{n_k} \}.
\]
Let $A$ be the set from \eqref{eqn: density(A)}, and define
\begin{multline*}
C_k = \{i \in I_k :\text{$i$ has only 0's and 1's in its ternary expansion} \\
    \text{and the 1's only occur at positions in $A - m_{k-1}$\}}.
\end{multline*}
By Lorentz's~Lemma, there exists a set $B_k \subseteq (-3^{n_k}, 3^{n_k})$ with
\begin{equation} \label{eq: lbleq1}
I_k \subseteq C_k + B_k
\end{equation}
satisfying, for all $\epsilon > 0$ and all sufficiently large $k$,
\begin{equation*}
\card B_k
 \leq c' \cdot 3^{n_k} \cdot \frac{\log \left[2^{n_k (D + \epsilon)} \right]}{ 2^{n_k[D - \epsilon]}}
% = c' \cdot 3^{n_k} \cdot \log 2 \cdot \frac{n_k(D+\epsilon)}{2^{n_k (D - \epsilon)}}
 = c \cdot 3^{n_k} \cdot \frac{n_k(D+\epsilon)}{2^{n_k (D - \epsilon)}}
\end{equation*}
where $c'$ is the constant obtained from Lorentz's~Lemma, $c = c' \log
2$, and again $D = \density(A)$.  Let
\[
\frac{1}{2}\C_A = \left\{ \sum_{k=1}^\infty  \frac{a_k}{3^{m_k}} : a_k \in C_k \right\} \quad\text{and}\quad 
\frac{1}{2}E = \left\{ \sum_{k=1}^\infty  \frac{b_k}{3^{m_k}} : b_k \in B_k \right\}.
\]
The set $E$ is closed since it is the countable intersection of closed sets.  Let $x \in [0,1]$ with ternary expansion $0.x_1 x_2 x_3 \dotsc$.  By \eqref{eq: lbleq1}, there exist sequences $\{a_k\}$ with members in $C_k$ and $\{b_k\}$ from $B_k$ such that for all $k$,
$
\sum_{j=0}^{n_k-1}  x_{m_k-j} 3^j = a_k + b_k,
$
and therefore
\[
x = \sum_{k=1}^\infty \sum_{j=m_{k-1}+1}^{m_k} \frac{x_j}{3^j} 
=  \sum_{k=1}^\infty \frac{1}{3^{m_k}}\sum_{j=0}^{n_k-1}x_{m_k-j}3^j
= \sum_{k=1}^\infty \frac{a_k}{3^{m_k}} +\sum_{k=1}^\infty \frac{b_k}{3^{m_k}}.
\]
is a member of $\in \frac{1}{2}\C_A + \frac{1}{2} E$.  This proves $\C_A + E = [0,2]$.

It remains to verify that $\cdimh E \leq \alpha$.  Let $\epsilon > 0$,
and let $x \in E$. We want to compute an upper bound on $K(x \restr m_k)$. To specify $x \restr m_k$, we can first specify the
sets $B_j$, for $j \leq k$ and then specify which element of
$B_1\times\dotsb\times B_k$ gives the blocks of $x \restr m_k$.

If we know $A[m_k]$, we can determine the sequence of sets $B_j$, for $j \leq k$
(just use a brute force search to find the first $B_j$ as in the
conclusion of Lorentz's Lemma); by Lemma~\ref{lem: K(A) is small} this
requires a ternary string of length at most $4\log m_k + O(1)$ (plus an additional $o(\log_3
m_k)$ for starting and ending delimiters if desired).
An element of the known set $B_1\times\dotsb\times B_k$ can be specified
by a ternary string of length at most
\begin{multline*} \label{sauce}
\log_3 \prod_{j=1}^k \card{B_j}
\leq \log_3 \prod_{j=1}^k c \cdot 3^{n_j} \cdot \frac{n_j (D+ \epsilon)}{2^{n_j(D - \epsilon)}}\\
= k\log_3 c + m_k + \log_3 \prod_{j=1}^k \frac{n_j (D+ \epsilon)}{2^{n_j(D - \epsilon)}}
 \leq   k\log_3 c + m_k + \log_3 \frac{(D+\epsilon)^k \cdot (m_k)^k}{2^{m_k \cdot (D-\epsilon)}}\\
 \leq k \log_3 c + k^2 + k\log_3(D+\epsilon) + 2k \log_3 k - k^2 (D-\epsilon) \dimC.
\end{multline*}
(and again we can add $O(\log_3 m_k)$ for delimiters).
Therefore,
\begin{align*}
K(x \restr m_k) &\leq k^2[1-(D-\epsilon)\dimC + o(1)] + 4 \log_3 m_k +
O(\log_3 m_k)
\\ &= k^2[\alpha + \epsilon \cdot \dimC + o(1)],
\end{align*}
and appealing to the Kolmogorov complexity definition for constructive dimension \eqref{spaghetti}, we find
\[
\cdim\{x\} 
\leq \liminf_{k \to \infty} \frac{K(x \restr m_k)}{m_k}
\leq \alpha + \epsilon \cdot \dimC.
\]
It follows from \eqref{meatball} that $\cdimh E \leq \alpha + \epsilon$ for every $\epsilon > 0$.
\end{proof}

Take $E$ as in Lemma~\ref{lem: lbt dim E lemma} and let $F = 2-E$.  Then $F \subseteq \Eleq$
and $F - \C_A =  2- (E+\C_A) = [0,2]$.
Fix $r \in [0,1]$ and let $S = \C \cap (F - r)$; it will suffice to show
that $\dimh S \ge \alpha - 1 + \dimC$.

Now for each $z \in \C$ there exist unique points $v \in \C_A$ and $w \in \C_{\bar{A}}$ such that $v+w = z$; let $p$ be the projection map which takes $z \in \C$ to its unique counterpart $w \in \C_{\bar{A}}$.  For each $y \in \C_{\bar{A}}$ we have $r + y \in [0,2] \subseteq F - \C_A$,
so there exists $x \in \C_A$ such that $r+y \in F - x$, which gives
$x+y \in S$ since $\C_A+\C_{\bar{A}}=\C$.    Thus $p$ maps $S$ onto $\C_{\bar{A}}$.  Since $p$ is Lipschitz we have, using \eqref{eqn: hello kitty},
\begin{equation} \label{eq: dim S}
\dimh S \geq \dimh \C_{\bar{A}} \geq \alpha - 1 + \dimC
\end{equation}
because Lipschitz maps do not increase dimension \citep[Lemma~1.8]{Fal86}.  Theorem~\ref{lbt} follows.
\end{proof}

\begin{rem}
The set $E$ constructed in Lemma~\ref{lem: lbt dim E lemma} has both Hausdorff dimension and effective Hausdorff dimension $\alpha$.  Following the method of \eqref{eqn: hello kitty}, we can establish the following lower bound:
\[
\dimh E \geq \dimh [0,2] - \dimbu \C_A \geq 1 + \alpha - 1 = \alpha.
\]
\end{rem}

\section{Upper bound} \label{sec: upper bound}

In this section we prove the following upper bound which matches the lower bound of Theorem~\ref{lbt} and Theorem~\ref{thm: Eeq lb}.

\begin{thm}\label{ubt}
Let $1 -\dimC \le \alpha \le 1$. For every  Martin-L\"{o}f random real $r$,
\begin{equation} \label{ubtneq}
\dimh \left[(\C+r) \cap \Eleq\right] \leq \alpha - 1 + \dimC.
\end{equation}
\end{thm}
\begin{proof}
The case $\alpha=1$ is trivial, so assume $\alpha<1$.
Fix a computable $\gamma > \alpha + \dimC \ge 1$, and let $t = \gamma - 1$. Let $f$ be defined as in \eqref{eq: defn of f} and, as before, let $B = f^{-1}(\C\times \Eleq)$, so that the vertical fiber of $B$ at $x$ is $B_x = (\C+x)\intersect \Eleq$.  To prove Theorem~\ref{ubt} it suffices to prove the following lemma.

\begin{lemma} \label{ubl1}
For every Martin-L\"{o}f random real $r$, $\Ha^t(B_r) = 0$.
\end{lemma}
Let $\M^t$ be the \emph{$t$-dimensional net measure} in the plane induced by the net of standard dyadic squares, and for each $\delta > 0$, let $\M^t_\delta$ be the $\delta$-approximate net measure \cite{Fal86}. $\M^t$ and $\M^t_\delta$ are defined in the same way  as $\Ha^t$ and $\Ha^t_\delta$ except that the covers $\G$ from the definition in \eqref{def: Ha} consist exclusively of square sets of the form
\[
\left[\frac{m}{2^k}, \frac{m+1}{2^k}\right) \times \left[\frac{n}{2^k}, \frac{n+1}{2^k}\right)
\]
for integers $k$, $m$, and $n$.  Hence for any set $E$, $\M^t_\delta(E) \geq \Ha^t_\delta(E)$.  Let $\Le =  \Ha^1$ be Lebesgue measure on the real line.  

To prove Lemma~\ref{ubl1}, it suffices to show the following which is a ``computable'' version of Marstand's method below.  Indeed any $x$ which belongs to the left-hand side of \eqref{eq: mardstrand test} for every $m$ fails a Martin-L\"{o}f test and therefore cannot be Martin-L\"{o}f random.
\begin{lemma} \label{ubl2}
Fix $c> 0$. There is a uniformly computable sequence of open sets $U_m$ with $\Le(U_m) < 2^{-m}$ such that for each $m$,
\begin{equation} \label{eq: mardstrand test}
\{x:\Ha^t(B_x) >c\}\subseteq U_m.
\end{equation}
\end{lemma}

\begin{proof}[Proof of Lemma~\ref{ubl2}]
We may assume that $c$ is computable.
Fix a uniformly computable sequence of collections $\S_k
  = \{S_{k,i}\}$ of dyadic squares
  which for each $k$, forms a $2^{1/2}\cdot\frac{1}{k}$-\emph{mesh cover} of $B$ with
\[
\sum_{S_{k,i} \in \S_k} \size{S_{k,i}}^\gamma < \frac{c \cdot 2^{1/2}}{2^{k+1}}.
\]
We show the existence of such a sequence in Lemma~\ref{lem: S_k squares}. For each $k$, let

\[
A_k = \left\{x: \sum_{i=0}^\infty \left|\left(S_{k,i}\right)_x\right|^t > c\right\}
\]
where $\left(S_{k,i}\right)_x$ denotes the vertical fiber of $S_{k,i}$ at $x$.
The sets $A_k$ are unions of left-closed right-open dyadic intervals
in a uniformly computable way.
Since
\[
\Ha^t(B_x) >c\ \implies\ x \in \bigcap_{m=1}^\infty\bigcup_{k=m}^\infty A_k,
\]
we shall see that it suffices to show $\Le(A_k) < 2^{-k-1}$.

Let $a_{k,i} = \size{S_{k,i}}^t$ and let
\[
I_{k,i} = \{x : (x,y) \in S_{k,i}\text{ for some $y$}\}.
\]
Each $I_{k,i}$ is a dyadic interval and $|I_{k,i}| = 2^{-1/2}
\size{S_{k,i}}$. Also, if $x \in A_k$, then
$
 \sum_{\{i : x\in I_{k,i}\}} a_{k,i} > c. 
$
We now appeal to \citep[Lemma~5.7]{Fal86}:
\begin{marstrandlemma}[\citep{Mar54}]  \label{lem:marstrand}
Let $A \subseteq \R$, let $\{I_n\}$ be a $\delta$-mesh cover of $A$ by dyadic intervals, and let $a_n > 0$ for all $n$.  Suppose that for all $x \in A$
\[
\sum_{\{n:x \in I_n\}} a_n > c
\]
for some constant $c$.  Then for all $s$,
\[
\sum_{n} a_n \size{I_n}^s \geq c \cdot \M_\delta^s(A).
\]
\end{marstrandlemma}
Applying Marstrand's~Lemma with $s= 1$, we have

\[
\frac{c}{2^{k+1}} > 2^{-1/2} \sum_{S_{k,i} \in \S_k} \size{S_{k,i}}^\gamma = \sum_{\{i: x\in I_{k,i}\}} a_{k,i} \size{I_{k,i}} \geq c \cdot \M_{1/k}^1(A_k).
\]

It follows that $\Le(A_k) < 2^{-k-1}$. Thus, for each $m$,
$\Le(\cup_{k=m}^\infty A_k) < 2^{-m}$. Now we may find uniformly
computable open sets $U_m$ with 
\[
\{x:\Ha^t(B_x) >c\} \subseteq \bigcup_{k=m}^\infty A_k \subseteq U_m
\]
and $\Le(U_m) < 2^{-m}$.
\end{proof}
For $s \in [0,1]$, a \emph{weak $s$-randomness test} \cite{Tad02} is a sequence of uniformly c.e.\ sets of open dyadic intervals $U_0, U_1, U_2, \dotsc$ such that $\sum_{\sigma \in U_n} 2^{-s\size{\sigma}} \leq 2^{-n}$ for all $n$.  We will call a set $E \subseteq \R$ \emph{weakly $s$-random} if $E \not\subseteq \bigcap_n U_n$ for every weak $s$-randomness test  $U_0, U_1, U_2, \dotsc$.  We will need the following technical result in order to ensure that the cover $\S_k$ in Lemma~\ref{lem: S_k squares} is sufficiently uniform:
\begin{lemma} \label{lem: S_k squares helper}
For every $d > \alpha$ there is a computable function
  $\pair{j,k,l} \mapsto Q_{\pair{k,l},j}$ such that for all $k$ and $l$,
\begin{enumerate}[\scshape (i)]
\item $\left\{Q_{\pair{k,l},j}\right\}_j$ is a $2^{-k}$-mesh cover of $\Eleq$ by dyadic intervals, and
\item $\sum_{j = 0}^\infty \size{Q_{\pair{k,l},j}}^d < 2^{-l}$.
\end{enumerate}
\end{lemma}

\begin{proof}[Proof of Lemma~\ref{lem: S_k squares helper}]
Without loss of generality, we can assume that $l \geq kd$; proving the result for a larger $l$ only makes the second part of the lemma more true.  An inspection of \cite[Proposition~13.5.3]{DH10} (and the proof of the theorem immediately preceding it) reveals that for any $s \in [0,1]$ and any set of reals $E$,
\[
\cdim E \geq \sup \{s : \text{$E$ is weakly $s$-random} \}.
%I omit the proof details in order to avoid introducing prefix-free complexity.
\]
Since $\cdimh(\Eleq) = \alpha < d$ \cite[Theorem~4.7]{Lut03}, we have that $\Eleq$ is not weakly $d$-random.  This means that there exists a uniformly c.e.\ collection of dyadic intervals $\Q_{k,l} = \{Q_{\pair{k,l},j} : j \geq 0\}$ such that
\begin{equation} \label{yogurt}
\sum_{j=0}^{\infty} \size{Q_{\pair{k,l},j}}^d \leq 2^{-l}
\end{equation}
and each $\Q_{k,l}$ covers $\Eleq$ which proves \textsc{(ii)}.  It follows from \eqref{yogurt} that for every $j_0$, $\size{Q_{\pair{k,l},j_0}}^d \leq 2^{-l}$, and so $$\size{Q_{\pair{k,l},j_0}} \leq 2^{-l/d} \leq 2^{-k}$$ as needed for \textsc{(i)}.
\end{proof}

\begin{lemma} \label{lem: S_k squares}
There exists a uniformly computable sequence of collections of dyadic squares $\S_k$
 which, for each $k$, form a $2^{1/2}\cdot\frac{1}{k}$-mesh cover of $B$ with
\begin{equation*} \label{eq: make 2^-l small}
\sum_{S \in \S_k} \size{S}^\gamma < \frac{c \cdot 2^{1/2}}{2^{k+1}}.
\end{equation*}
\end{lemma}

\begin{proof}[Proof of Lemma~\ref{lem: S_k squares}] 
Write $\gamma =  s+d$, where $s > \dimC$ and $d > \alpha$.
Let $\G_k$ be a uniformly computable mesh cover of $\C$ by dyadic intervals of length $2^{-k}$ such that for each $k$, $\card (\G_k) \leq 2^{sk}$ and, let $\Q_{k,l} = \{Q_{\pair{k,l},j} : j \geq 0\}$ as in Lemma~\ref{lem: S_k squares helper}.  Form the uniformly computable sequence of square covers:
\[
\Gamma_{k,l} = \{G \times Q : \text{$Q \in \Q_{k,l}$ and $G \in \G_{-\log |Q|}$}\}.
\]
Then, using $\card[\G_{-\log |Q|}] \leq \size{Q}^{-s}$ for all $Q \in \Q_{k,l}$,
\begin{align*}
\sum_{X \in \Gamma_{k,l}} \size{X}^\gamma 
= \sum_{Q \in \Q_{k,l} \atop G \in \G_{-\log |Q|}} \size{G \times Q}^\gamma 
&= 2^{\gamma/2}\sum_{Q \in \Q_{k,l}} \size{Q}^{s+d} \cdot \card \left[\G_{-\log |Q|}\right] \\
&\leq 2^{\gamma/2} \sum_{Q \in \Q_{k,l}}\size{Q}^d
< 2^{-l}.
\end{align*}
Let $f$ be the Lipschitz mapping $\eqref{eq: defn of f}$ whose inverse map does not increase diameter by more than a factor of $\sqrt{2}$ and maps $\Gamma_{k,l}$ onto $B$ for all $k$ and $l$.  For every $k$, let $m(k)$ be sufficiently large so that $2^{-m(k)}$ is less than $c / 2^{k+1}$.  Now
form the collection $\S_k$ by taking, for each $X \in \Gamma_{k,m(k)}$, the two dyadic
squares which together cover the sheared dyadic square $f^{-1}(X)$.
%\[
%\S_k = \left\{f^{-1}(X) : X \in \Gamma_{k,m(k)} \right\}
%\] 
Then the $\S_k$'s form a uniformly computable sequence of square covers which achieves the desired bounds.
 \end{proof}
This concludes the proof of Theorem~\ref{ubt}.
\end{proof}

\begin{rem}
In contrast to the lower bound in Theorem \ref{lbt}  which holds for
\emph{all} reals in $[0,1]$, the upper bound in Theorem~\ref{ubt} indeed requires some hypothesis on $r$.  Indeed if $r=0$ and $\dimC < \alpha < 1$ satisfied \eqref{ubtneq}, we would have
\[
\dimC  = \dimh \left[(\C + 0) \intersect \Eleq\right] \leq \alpha - 1 + \dimC < \dimC,
\]
a contradiction.
\end{rem}

\section{Lower bound II} \label{sec: lb2}

We modify the proof of Theorem~\ref{lbt} to obtain a stronger result for the case of Martin-L\"{o}f randoms:
\begin{thm} \label{thm: Eeq lb}
Let $1 - \dimC \leq \alpha \leq 1$ and let $r \in [0,1]$ be
Martin-L\"{o}f random. Then
\[
\dimh \left[(\C+r) \intersect \Eeq\right] = \alpha - 1 + \dimC.
\]
Moreover,
\[
\Ha^{\alpha - 1 + \dimC} \left[(\C+r) \intersect \Eeq \right] > 0.
\]

\end{thm}
\begin{proof}
Fix an  $\alpha$ satisfying $1-\dimC \leq \alpha \leq 1$, and let $A$, $\C_A$, $\C_{\bar{A}},
E, F$, and $S$ be as in the proofs of Theorem~\ref{lbt} and Lemma~\ref{lem: lbt dim E lemma}.  For $x \in \R$, let 
\[
N_\delta(x) = \{y \in \R : \size{x-y} \leq \delta\}.
\]
We shall make use of the following result from Mattila's book:
\begin{thm}[\cite{Mat95}, Theorem~6.9] \label{thm: mattilathm}
Let $\mu$ be a Radon measure on $\R^n$, $E \subseteq \R^n$, $0 < \lambda < \infty$, and $\alpha > 0$.  If 
\[
\limsup_{\delta \to 0} \frac{\mu[N_\delta(x)]}{(2\delta)^{\alpha}}  \leq \lambda
\]
for all $x \in E$, then $\Ha^\alpha(E) \geq \frac{\mu (E)}{2^\alpha \lambda}$.
\end{thm}

\begin{lemma} \label{lem: C_A has positive measure}
$\Ha^{\alpha - 1 + \dimC}(\C_{\bar{A}}) > 0$.
\end{lemma}
\begin{proof}[Proof of Lemma~\ref{lem: C_A has positive measure}]
Since the case $\alpha = 1-\dimC$ is trivial, we assume $\alpha > 1-\dimC$.
Let $\beta  = \alpha-1+\dimC$.  Since $A =  \{\lfloor y/D \rfloor :y \in \N\}$ where 
$D =  \frac{1-\alpha}{\dimC}$, we have $\overline{A}
= \{u_1 < u_2 <u_3 <\ldots\}$, where $\lim_{n\to \infty} \frac{n}{u_n} =
1-D$.
In fact, the careful choice of the set $A$ lets us make a stronger statement
about the numbers $u_n$: there is a fixed number $t$ such that
$u_n \le (n+t)/(1-D)$ for all $n$.  For each finite binary string $\sigma = \sigma_1 \sigma_2 \dotsc \sigma_n$, let $I(\sigma)$ be the closed interval of ternary expansions $x = 0.x_1 x_2 \dotsc$ satisfying
\[
x_p = 
\begin{cases}
\sigma_k & \text{if $p = u_k$ for some $k \leq n$, and} \\
0 & \text{if $p \leq u_n$ and $p \notin \{u_1, \dotsc, u_n\}$}. \\
\end{cases}
\]
Define a probability measure $\mu$ on $[0,1]$ by requiring
\[
\mu[I(\sigma) \intersect \C_{\bar{A}}] = \frac{1}{2^{[\text{length of $\sigma$}]}}.
\]
Since $\mu$ is a bounded Borel measure supported on the compact set $\C_{\bar{A}}$, $\mu$ is a Radon measure; hence Theorem~\ref{thm: mattilathm} applies.  For $\delta > 0$, let $f(\delta)$ be the least index such that $\delta > 3^{-u_{f(\delta)}}$.  
\begin{comment}
TAKE THIS OUT!
For $x \in \frac{1}{2} \C_{\bar{A}}$, we calculate:
\begin{gather*}
\limsup_{\delta \to 0} \frac{\mu[N_\delta(x)]}{(2\delta)^\beta} 
\leq \limsup_{\delta \to 0} \frac{2^{-f(\delta)}}{\left(2 \cdot 3^{-u_{f(\delta)}} \right)^\beta}
\leq \limsup_{\delta \to 0} \frac{3^{\beta \cdot \frac{f(\delta) + 1}{c}}}{2^\beta \cdot 2^{f(\delta)}} \\
= \limsup_{\delta \to 0} \frac{3^{\beta/c}}{2^\beta} \cdot \left(\frac{3^{\beta/c}}{2}\right)^{f(\delta)}
= \frac{3^{\beta/c}}{2^\beta}.
\end{gather*}
\end{comment}
Let
$$
x = .000 \dotsb 00x_{u_1}000 \dotsb 00x_{u_2}000 \dotsb 00x_{u_f(\delta)}000\dotsb \in \frac{1}{2} \C_{\bar{A}},
$$
and let $J = I[x_{u_1}x_{u_2} \dotsb x_{u_{f(\delta)}}]$.  Then $x \in J$ and the length of the closed interval $J$ is $3^{-u_{f(\delta)}} < \delta$. So $J \subseteq N_\delta(x)$.  Now the length of each interval $I[\sigma_1 \dotsb\sigma_{f(\delta)-1}]$ is $3^{-u_{f(\delta)-1}} \geq \delta$, so $N_\delta(x)$ can intersect no more than 4 of these non-overlapping intervals. Therefore $\mu[N_\delta(x)] \leq 4\cdot 2^{-(f(\delta)-1)}$.  Let $c = 1-D$.  Then
$$
\frac{\mu[N_\delta(x)]}{(2\delta)^\beta} 
\leq  \frac{4\cdot 2^{-(f(\delta)-1)}}{(2\cdot 3^{-u_{f(\delta)}})^\beta} 
\leq \frac{8}{2^\beta} \cdot \left(\frac{3^{\beta u_{f(\delta)}}}{2^{f(\delta)}}\right)
\leq \frac{8}{2^\beta} \cdot \left(\frac{3^{\beta \cdot \frac{f(\delta)+t}{c}}}{2^{f(\delta)}}\right).
$$
Thus
$$
\limsup_{\delta \to 0} \frac{\mu[N_\delta(x)]}{(2\delta)^\beta}
\leq \limsup_{\delta \to 0} 8 \cdot \frac{3^{\beta t/c}}{2^\beta} \cdot \left(\frac{3^{\beta / c}}{2} \right)^{f(\delta)}
= 8 \cdot \frac{3^{\beta t/c}}{2^\beta},
$$ 
and hence by Theorem~\ref{thm: mattilathm}, we have $\Ha^\beta\left(\frac{1}{2}\C_{\bar{A}}\right) \geq \frac{3^{-\beta t/c}}{8}$.  It follows that $\Ha^{\alpha-1+\dimC}(\C_{\bar{A}}) = \Ha^\beta(\C_{\bar{A}}) > 0$.
\end{proof}

\begin{lemma} \label{ewok}
$\Ha^{\alpha - 1 + \dimC} \left[(\C+r) \intersect \Eleq \right] > 0$.
\end{lemma}

\begin{proof}[Proof of Lemma~\ref{ewok}]
Let $S = \C \intersect (E-r)$.  By Lemma~\ref{lem: lbt dim E lemma}, $E \subseteq \Eleq$, and so it suffices to show that $\Ha^{\alpha - 1 + \dimC} (S) > 0$.  Retracing the argument of Theorem~\ref{lbt} down to \eqref{eq: dim S}, we get $\dimh S \geq \alpha - 1 + \dimC$.  Furthermore, as we now argue,
\begin{equation} \label{eq: Ha S}
\Ha^{\alpha-1 + \dimC}(S) \geq q_\alpha \cdot \Ha^{\alpha-1+\dimC}(\C_{\bar{A}}) > 0
\end{equation}
for some constant $q_\alpha > 0$.  The strict inequality in \eqref{eq: Ha S} follows from Lemma~\ref{lem: C_A has positive measure}.  For the nonstrict inequality, we again appeal to \cite[Lemma~1.8]{Fal86} and the fact that the projection map from $\C$ to $\C_A$ is Lipschitz.  \cite[Lemma~1.8]{Fal86} states that, up to some constant factor, a Lipschitz map does not decrease $\Ha^\alpha$ measure.  This is slightly stronger than what we used before in Theorem~\ref{lbt}, namely that a Lipschitz map cannot increase dimension.
\end{proof}

Using the assumption that $r$ is Martin-L\"{o}f random, we next obtain the following:
\begin{lemma} \label{jedi}
$\Ha^{\alpha - 1 + \dimC} \left[(\C+r) \intersect \Elt \right] = 0$.
\end{lemma}

\begin{proof}[Proof of Lemma~\ref{jedi}]
The case $\alpha=1-\dimC$ follows from Theorem~\ref{simp}, so assume $\alpha>1-\dimC$.
By Theorem~\ref{ubt}, for any $\gamma$ with $1 - \dimC \leq \gamma < \alpha$, we have
\[
\dimh \left[ (\C + r) \intersect {\mathcal E_{\leq\gamma}} \right] \leq \gamma - 1 + \dimC < \alpha - 1 + \dimC.
\]
Since $\Elt$ is the countable union of sets $E_{\leq\gamma}$ for a sequence of $\gamma$'s
approaching $\alpha$ from below, the theorem follows.
\end{proof}
Combining Lemma~\ref{ewok} with Lemma~\ref{jedi}, we find that
\[
\Ha^{\alpha - 1 + \dimC} \left[(\C+r) \intersect \Eeq \right] = \Ha^{\alpha - 1 + \dimC} \left[(\C+r) \intersect (\Eleq - \Elt) \right] > 0,
\]
whence we conclude the desired theorem.
\end{proof}
Combining Theorem~\ref{thm: Eeq lb} with Theorem~\ref{lbt} and Theorem~\ref{ubt} yields our final result:
\begin{cor} \label{cor:dim=dim}
For any $\alpha$ satisfying $1 - \dimC \leq \alpha \leq 1$, and for any Martin-L\"{o}f random $r \in [0,1]$, we have
\begin{equation*} \label{eqn:C+r and Eleq}
\dimh \left[(\C+r) \intersect \Eeq\right] = \dimh \left[ (\C + r) \intersect \Eleq \right] = \alpha - 1 + \dimC.
\end{equation*}
\end{cor}

\bibliographystyle{amsplain}
\bibliography{cantor_set}
\end{document}